\documentclass[11pt]{article}
\usepackage[margin=1in]{geometry}
\usepackage{hyperref}
\usepackage{amsfonts}
\usepackage{amsmath}
\usepackage{enumerate}
\usepackage{amsthm}
\usepackage{graphicx}
\usepackage{todonotes}
\usepackage{microtype} 
\usepackage{cite}
\graphicspath{{figures/}}

\newtheorem{problem}{Problem}
\newtheorem{theorem}{Theorem}
\newtheorem{definition}{Definition}
\newtheorem{lemma}{Lemma}
\newcommand{\prob}[1]{\textsc{#1}}
\newcommand{\cc}[1]{\textup{\textsf{#1}}}
\DeclareMathOperator{\poly}{poly}
\DeclareMathOperator{\snub}{Snub}
\DeclareMathOperator{\trifreeergm}{TriFreeERGM}

\title{ERGMs are Hard}
\date{}
\usepackage{authblk}
\author{Michael J. Bannister}
\author{William E. Devanny}
\author{David Eppstein}
\affil{Department of Computer Science, University of California, Irvine}

\begin{document}
\maketitle
\thispagestyle{empty}
\begin{abstract}
We investigate the computational complexity of the exponential random graph model (ERGM) commonly used in social network analysis. This model represents a probability distribution on graphs by setting the log-likelihood of generating a graph to be a weighted sum of feature counts. These log-likelihoods must be exponentiated and then normalized to produce probabilities, and the normalizing constant is called the \emph{partition function}. We show that the problem of computing the partition function is $\cc{\#P}$-hard, and inapproximable in polynomial time to within an exponential ratio, assuming $\cc{P} \neq \cc{NP}$. Furthermore, there is no randomized polynomial time algorithm for generating random graphs whose distribution is within total variation distance $1-o(1)$ of a given ERGM. Our proofs use standard feature types based on the sociological theories of assortative mixing and triadic closure.
\end{abstract}

\newpage
\setcounter{page}{1}
\section{Introduction}

An \emph{exponential random graph model} (ERGM) is a mathematical description of a probability distribution over the graphs with a given fixed set of vertices, used in social network analysis~\cite{WasPat-Psy-96,PatWas-BJMSP-99,RobPatWas-Psy-99,HunHanBut-JSS-08}. In this model, every graph is mapped to a \emph{feature vector} (typically a Boolean or integer vector describing the existence or number of local structures such as vertices with given degrees or small induced subgraphs of a given type) and the log-likelihood of each graph is determined as the inner product of its feature vector with a weight vector specified as part of the model.  ERGMs were first developed by Holland and Leinhardt to study graphs arising from social relationships~\cite{HolLei-JASA-81}. In contrast to other probabilistic models for social networks, such as the uniform distribution on graphs with a given degree distribution~\cite{ChuLu-AoC-02}, the features and weights of an ERGM can correspond directly to sociological  theories of network formation. For instance, features that describe the presence or absence of individual edges may have weights derived from models of \emph{assortative mixing}, the theory that people with similar characteristics are more likely to form connections with each other~\cite{New-PRE-03}, while features describing the presence of triangle subgraphs may have weights derived from models of \emph{triadic closure}, the theory that friends-of-friends are likely to become connected~\cite{Dav-ASR-70}. By using machine learning algorithms to fit the weights of an ERGM to real-world data, sociologists may experimentally measure the strength of these effects and use them to test their theories.

Although there has been some theoretical research on speeding up the computation of feature vectors~\cite{EppSpi-JGAA-12,EppGooStr-TCS-12}, a low-level step in ERGM computations,  as well as research on ERGMs from the graph limit point of view~\cite{Lov-LNGL-12}, little has been known to date about the higher level computational complexity of these models. In order to solve key problems on ERGMs, including the problems of generating graphs from a given model, computing the \emph{partition function} of a model (a normalizing constant used to transform log-likelihoods into probabilities), and fitting weights to data, researchers have generally resorted to Monte Carlo methods, in which the Metropolis--Hastings algorithm is used to construct a Markov chain on the set of graphs defined by an ERGM, with its stable distribution equal to the distribution described by the ERGM. Then, this chain is simulated for a number of steps with heuristic termination conditions that are intended to detect convergence of the chain to its stable distribution~\cite{FraStr-JASA-86,HunHan-JCGS-06}. However, these methods have no guarantees on their running time, accuracy, mixing time, rate of convergence, or correctness of the termination detection method.

In this paper, for the first time, we investigate ERGMs from the point of view of their computational complexity. We explain the heuristic nature of previous computations involving these models by showing that several key computational problems involving ERGMs are intractable in the worst case. In particular, we show, for a family of ERGMs parameterized by the number $n$ of vertices whose description complexity (number of features and weights, and magnitude of the weights) is polynomial in $n$, that:

\begin{itemize}
\item Unless $\cc{P}=\cc{NP}$, there is no polynomial-time approximation to the partition function of a given ERGM that can achieve an approximation ratio exponential in any polynomial of~$n$.
\item Unless $\cc{P}=\cc{\#P}$, there is no polynomial-time algorithm to compute the partition function of a given ERGM.
\item Unless $\cc{RP}=\cc{NP}$, there is no randomized polynomial algorithm for generating random graphs whose output distribution is within total variation distance $1-o(1)$ of a given ERGM with variable weights.
\end{itemize}

Our results can be obtained using an ERGM with features that are very natural in social network analysis: an independent weight for each potential edge in a graph (representing different affinities between different pairs of actors in a network) and a single shared weight for all induced triangles (representing triad closure). As we show, an ERGM of this type can be used to describe a distribution that is very close to the uniform distribution on the maximum triangle-free subgraphs of a given graph. Our results follow from the known $\cc{NP}$-hardness of finding a maximum triangle-free subgraph~\cite{Yan-SJoC-81} and from a new $\cc{\#P}$-hardness proof for counting large triangle-free subgraphs. We also show that these problems remain equally hard for an ERGM in which the features are the induced subgraphs isomorphic to $H$, for every fixed graph $H$ with three or more vertices, so small adjustments to the types of features available in the model cannot make these problems easier. Thus, our results destroy all hope of guaranteed-quality polynomial-time computation for ERGM models of social networks.

From the point of view of theoretical computer science, our methods are mostly standard reductions. However, in this respect, our most innovative result may be the inapproximability of the ERGM partition function. There have been past results on $\cc{\#P}$-hardness and inapproximability of partition functions~\cite{Bar-JPMG-82} but these have generally involved systems of colorings or other decorations on a fixed underlying graph; instead, the states of an ERGM are themselves all possible graphs on a given vertex set.
Our inapproximability result is much stronger than the polynomial inapproximability known to hold for all $\cc{\#P}$ problems~\cite{JerValVaz-TCS-86}. The proof of this result avoids PCP theory, which has been applied in many recent inapproximability results~\cite{AroBar-CCamm-09}, and instead proceeds by a direct reduction.

\section{Preliminaries}

\paragraph{Exponential random graphs.}
An exponential family random graph model, or ERGM for short, is a distribution on random graphs that forms an exponential family.  The distributions of exponential families have the form
\[
Pr[x] = \frac{e^{\theta \cdot w(x)}}{Z(\theta)}
\]
where $\theta$ is a vector of weights, sometimes called the model parameters, $w(x)$ is a vector valued function that gives features of $x$, and $Z(\theta)$, often called the partition function, is a normalizing value chosen to make the probabilities sum to one.  For an ERGM, $x$ would be a graph on $n$ vertices and the features of $x$ would often include localized structures such as vertices with a given degree, paths with a given length, or cliques with a given number of vertices.

We will let $F$ be the set of features we use and $w_f$ be the weight corresponding to a feature $f \in F$.  It will be convenient for us to linearly transform the weights (by multiplying by $\log_2 e$), so that we can use $2$ rather than $e$ as the base of our exponential functions, allowing us to avoid real number computations and simplify our proofs. We define the density of a graph $G$ to be $d(G) = \prod_{f \in F} 2^{f(G)w_f}$ where $f(G)$ is the value of the feature $f$ in $G$.  The features we use will be $0$-$1$ indicators of the presence of particular induced subgraphs and counts of subgraph appearances.  The probability of a graph, $G$, is its density, $d(G)$, divided by a normalizing constant to make the probabilities sum to one.  This normalizing constant is $Z = \sum_G d(G)$, so $Pr[G] = \frac{d(G)}{Z}$; $Z$ is called the partition function of an ERGM.  In this paper we are chiefly concerned with two computational problems: calculating $Z$ and generating graphs from these distributions.

The types of features used in ERGMs are diverse.  In general features can be of two types, homogeneous or heterogeneous, according to whether or not all isomorphic graphs $G$ have the same value of $f(G)$ as each other.  An example of a homogeneous subgraph feature would be the number of triangles in a graph.  An example of a heterogeneous subgraph feature would be a $0$-$1$ indicator that is one when three particular vertices form a triangle and zero otherwise. A collection of heterogeneous triangle features, one for each triple of vertices with all weights equal, are equivalent to a single homogeneous triangle feature with the same weighting.  It is common in social network analysis to weight edge features in a heterogeneous way: in these analyses, vertices typically represent people for whom we might have some prior knowledge or distribution of information that would affect their likelihood of being related. As an extreme example, in sexual contact networks, heterosexual contacts would be expected to be more frequent than homosexual contacts, and individuals who participate in both types of contact might be even less frequent.  In some sociological models, there is a common global constraint on the likelihood of triadic closures, or triangle subgraphs.  To align our proofs with this common model, we will use heterogeneous variable weights on edge features, but use a homogeneous uniform weight on all triangle subgraphs.

\newif\ifcomptheory
\comptheoryfalse
\ifcomptheory
\paragraph{Complexity theory.}
To discuss the difficulty of computing $Z$, we introduce some standard computational complexity background.

\begin{definition}
A \emph{problem} is a functional specification of an output for each input from a set of possible inputs and outputs, independent of the algorithm used to compute it.
\end{definition}

\begin{definition}
A \emph{decision problem} is a problem in which the output set is yes or no.
\end{definition}

\begin{definition}
A \emph{counting problem} is a problem in which the output set is the set of natural numbers.  Computing the number of Hamiltonian cycles in a graph is a counting problem.
\end{definition}

\prob{HAMILTONIAN}, the problem of determining if a graph has a Hamiltonian cycle, is a decision problem.  \prob{\#HAMILTONIAN}, the problem of counting the number of Hamiltonian cycles in a graph, is a counting problem.

\begin{definition}
A \emph{complexity class} is a set of problems that can be solved with a given limit on computational resources.
\end{definition}

\begin{definition}
\cc{P} is the complexity class of decision problems that can be solved in time polynomial in the length of the input.
\end{definition}

\begin{definition}
\cc{NP} is the complexity class of decision problems whose answers have polynomial time verifiable proofs.
\end{definition}

\prob{MATCHING}, the problem of determining if there is a perfect matching in a graph, has a polynomial time algorithm and so is in \cc{P}.  \prob{HAMILTONIAN} is in \cc{NP} because given an ordering of the vertices it can be checked in polynomial time if they form a Hamiltonian cycle in that order.  All problems in \cc{P} are trivially in \cc{NP} using a blank proof and their normal polynomial time algorithm.  The famous \cc{P}~vs.~\cc{NP} problem is to determine if these two complexity classes are the same.

\begin{definition}
\cc{\#P}, pronounced ``sharp p'', is the set of counting problems whose decision version is in \cc{NP}.
\end{definition}

Because \prob{HAMILTONIAN} is in \cc{NP}, \prob{\#HAMILTONIAN} is in \cc{\#P}.  In a sense \cc{NP} is the set of problems that ask about existence of solutions while \cc{\#P} is the set of problems that ask how many solutions a problem has.

\begin{definition}
We say a decision problem, \prob{A}, can be polynomial time reduced to another decision problem, \prob{B}, if there is an algorithm that in polynomial time maps instances of \prob{A} to instances of \prob{B} such that yes instances map to yes instances and no instances to no instances.  For counting problems, \prob{A} is reducible to \prob{B} if there is a pair of polynomial time algorithms such that the first one maps an instance of \prob{A} to an instance of \prob{B} such that the second algorithm applied to the answer to the \prob{B} instance gives the answer for the \prob{A} instance.
\end{definition}

\begin{definition}
A  problem is \cc{C}-hard, for a complexity class \cc{C}, if every problem in \cc{C} can be polynomial time reduced to that problem.
A problem is \cc{C}-complete if it is both in \cc{C} and it is \cc{C}-hard.
\end{definition}

Problems in \cc{P} are generally thought of as being tractable or efficiently solvable whereas problems not in \cc{P} are considered intractable because they take super-polynomial time to solve.  Because most researchers believe that \cc{P} is not equal to \cc{NP}, it is generally thought that \cc{NP}-hard problems do not have polynomial time algorithms.  Similarly \cc{\#P}-hard problems are also thought to take super-polynomial time.  In this paper we show different problems relating to ERGM computation are \cc{NP}-hard and \cc{\#P}-hard.
\fi

\paragraph{A menagerie of computational problems.}
We now introduce several computational problems that we will use throughout the paper:

\begin{problem} \prob{ERGM-PARTITION}

\textsc{Input:} an ERGM $E$

\textsc{Output:} the partition function for $E$
\end{problem}

\begin{problem} \prob{TRI-FREE}

\textsc{Input:} a graph $G=(V,E)$ and a positive integer $k$

\textsc{Output:} yes if there is a triangle free subgraph of $G$ with $k$ or more edges, no otherwise
\end{problem}

\begin{problem} \prob{\#TRI-FREE}

\textsc{Input:} a graph $G=(V,E)$ and a positive integer $k$

\textsc{Output:} the number of triangle-free subgraphs of $G$ with $k$ or more edges
\end{problem}

\begin{problem} \prob{\#MATCHING}

\textsc{Input:} a graph $G=(V,E)$

\textsc{Output:} the number of perfect matchings in $G$
\end{problem}

\begin{problem} \prob{\#3REG-BIP-MATCH}

\textsc{Input:} a $3$-regular bipartite graph $G=(V,E)$

\textsc{Output:} the number of perfect matchings in $G$
\end{problem}

\prob{\#MATCHING} and \prob{\#3REG-BIP-MATCH} are both known to be \cc{\#P}-complete~\cite{Val-TCS-79,DagLub-TCS-92}.   \prob{TRI-FREE} is \cc{NP}-complete~\cite{Yan-SJoC-81}.  We will prove that \prob{\#TRI-FREE} is \cc{\#P}-complete.

\section{Inapproximability}

Here we prove an inapproximability result on computing the partition function of an ERGM by reducing to \prob{TRI-FREE} and creating a gap in the values of the ERGM's partition function that separates the yes-instances from the no-instances.

\begin{definition}
For a given graph $G$, we define $\trifreeergm(G, \alpha)$ to be an ERGM with the following features and weights. We place heterogeneous weights on edges and a homogeneous weight on triangle subgraphs.  If an edge belongs to $G$ we give its indicator feature a weight of $\alpha$ and if it does not belong to $G$ we instead give it a weight of $\beta = -{n \choose 2}\alpha -{n\choose 2} - 1$.  Also, we assign the weight for the count of triangle subgraphs to be $\beta$.
\end{definition}

For two graphs $G$ and $H$ on the same vertex set, let $a$ be the number of edges in $H$ that are also in $G$, $b$ be the number of edges in $H$ that are not in $G$, and $c$ be the number of triangle subgraphs in $H$.  Then $d(H)$ in the distribution defined by $\trifreeergm(G, \alpha)$ is $2^{a\alpha + (b+c)\beta}$.  Also note that $a \leq {n\choose 2}$.

\begin{lemma}\label{lem:ergm-digits}
Fix an integer $\alpha > {n\choose 2}$ and a graph $G$, and let $d_i$ denote the number of triangle free subgraphs of $G$ with $i$ edges. Then the integer part of the partition function for $\trifreeergm(G, \alpha)$ can be rewritten in base $2^\alpha$ as $d_{n\choose 2} \dots d_2d_1d_0$.
\end{lemma}
\begin{proof}
If $H$ is any graph that has an edge not belonging to $G$ or that has a triangle subgraph, then its density $d(H)$ has a corresponding factor of $2^\beta$, and is thus strictly less than $1/2^{n\choose 2}$.  Therefore the sum total of the densities of all graphs that contain edges not in $G$ or that contain triangles is strictly less than $1$.

The remaining contributions to the partition function come from graphs that are triangle free subgraphs of $G$. Let $H$ be such a graph with $i$ edges; then $H$ contributes one unit to $d_i$ and $d(H) = 2^{i\alpha}$.  Therefore the integer part of $Z$ is $\lfloor Z \rfloor = \sum_i d_i2^{i\alpha}$.  Because each $d_i$ is less than $2^{n \choose 2}$ and $\alpha > {n\choose 2}$, there are no carries in this sum and $Z$ written in base $2^{\alpha}$ is $d_{n\choose 2} \dots d_2d_1d_0$.
\end{proof}

\begin{theorem}\label{thm:ergm-inapprox}
If $\cc{P} \neq \cc{NP}$, then \prob{ERGM-PARTITION} cannot be approximated within a factor of $2^{\poly(n)}$ in polynomial time.
\end{theorem}
\begin{proof}
We prove this by contrapositive and so suppose we could approximate the partition function of an ERGM within a factor of $f(n)$ in polynomial time where $f(n) = 2^{\poly(n)}$.

Let $G$ be the input to \prob{$k$-TRI-FREE}.

We use $E = \trifreeergm(G, \alpha)$.  If there is a triangle free subgraph with at least $k$ edges, then by \autoref{lem:ergm-digits} $Z > 2^{k\alpha}$.  If there are no triangle free subgraphs with $k$ edges, then by \autoref{lem:ergm-digits} $Z < 2^{{n \choose 2}}2^{(k-1)\alpha}$.  Therefore setting $\alpha = {n \choose 2} + 2\log \left(f(n)\right) + 1$, we have:
\begin{multline*}
f(n)^22^{{n \choose 2} + (k-1)\alpha} = 2^{2\log(f(n)) +  {n \choose 2} + (k-1)\alpha}\\
  = 2^{2\log(f(n)) +  k{n \choose 2} + 2(k-1)\log \left(f(n)\right) + k-1}\\
  = 2^{k\left( {n \choose 2} + 2\log \left(f(n)\right) + 1\right) -1}
  = 2^{k\alpha-1}
  < 2^{k\alpha}
\end{multline*}
So if the computed approximation of $Z$ is greater than $2^{k\alpha}/f(n)$, then $G$ necessarily has a $k$-edge triangle free subgraph. When the approximation of $Z$ is less than $f(n)2^{{n \choose 2} + (k-1)\alpha}$, $G$ has no $k$ edge triangle free subgraph.  Therefore if we could approximate $Z$ within a factor of $2^{\poly(n)}$ in polynomial time, then we could solve \prob{TRI-FREE} in polynomial time implying that $\cc{P} = \cc{NP}$.
\end{proof}

\section{Hardness}

In this section we prove the \cc{\#P}-hardness of computing the ERGM partition function.  We show this by reducing from \prob{\#TRI-FREE}. Unfortunately, the known reduction that shows \prob{TRI-FREE} to be \cc{NP}-hard is not parsimonious~\cite{Yan-SJoC-81}, meaning that the reduction does not preserve solution counts in a consistent way.  So we first need to show that \prob{\#TRI-FREE} is \cc{\#P}-complete, which we do by another reduction from \prob{\#3REG-BIP-MATCH}.

\begin{figure*}[t]
\centering\includegraphics[width=4in]{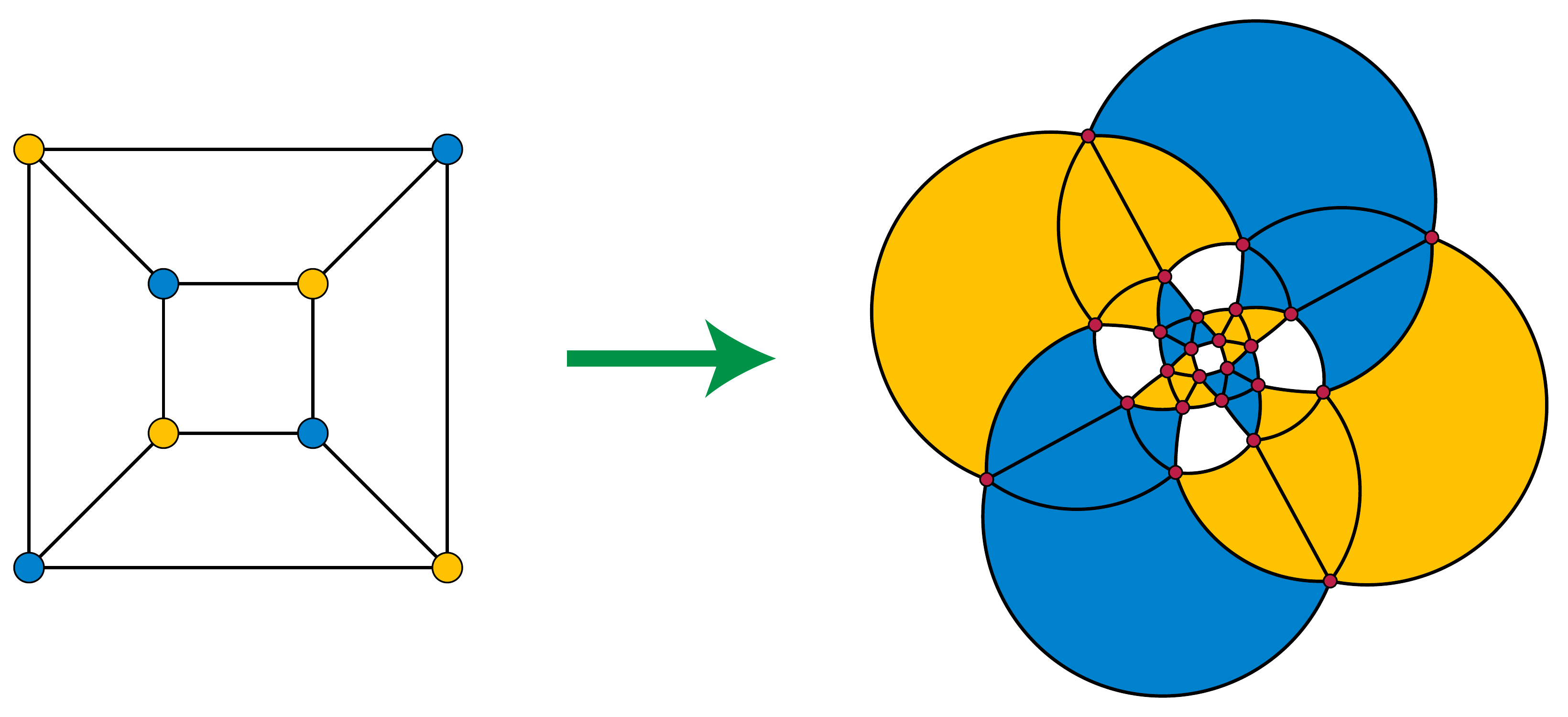}
\caption{Parsimonious reduction from \prob{\#3REG-BIP-MATCH} on a cube to \prob{\#TRI-FREE} on a snub cube}
\label{fig:match2maxtf}
\end{figure*}

\begin{definition}
We define \emph{$\snub(G)$} for a $3$-regular bipartite graph, $G$, to be another graph constructed from $G$ as follows.  For each vertex $v \in V(G)$, create a triangle $t_v$ in $\snub(G)$.  Call the edges of these triangles the vertex triangle edges.  Set an arbitrary cyclic ordering on the vertices of each of these triangles.  Then for each edge $(i,j) \in E(G)$, arbitrarily pick one unpicked vertex $u$ in $t_i$ and another unpicked vertex $w$ in $t_j$.  Add three edges to $\snub(G)$: one edge from $u$ to $w$, one edge from $u$ to the vertex after $w$ in $t_j$'s cyclic ordering, and one edge from $w$ to the vertex after $u$ in $t_i$'s cyclic ordering.  Call the edge from $u$ to $w$ the cross edge for $t_i$ and $t_j$ and call the other two edges connecting edges.
\end{definition}

The name of this construction comes from the snub cube and snub dodecahedron, two convex polyhedra whose graphs can be formed by applying this construction to the graphs of the cube and dodecahedron, respectively.
$\snub (G)$ has three edges for every vertex in $G$ and three more edges for every edge in $G$.  It also has one triangle for every vertex in $G$ and two triangles for every edge in $G$.  If $G$ has $n$ vertices, $\snub(G)$ has $15n/2$ edges and $4n$ triangles.

\autoref{fig:match2maxtf} illustrates the reduction, applied to the graph of a cube. In the figure, the vertex triangle for each cube vertex~$v$ corresponds to the central triangle in one of the patches of four triangles with the same color as~$v$.

\begin{lemma}\label{lem:tri-free-hard}
\prob{\#TRI-FREE} is \cc{\#P}-complete.
\end{lemma}
\begin{proof}
First we observe that \prob{\#TRI-FREE} $\in$ \cc{\#P}.

We will prove \cc{\#P}-hardness by reduction from \prob{\#3REG-BIP-MATCH}.

Given $G$, a $3$-regular bipartite graph on $n$ vertices, we reduce it to $\snub (G)$. The cross edges and vertex triangle edges of $\snub(G)$ each participate in two triangles, and the connecting edges each participate in only one triangle. However, the whole graph has $4n$ triangles. So, if a triangle-free subgraph of $\snub(G)$ can be obtained by deleting $2n$ edges from $\snub(G)$, leaving $11n/2$ edges behind, then we must only delete cross edges and vertex triangle edges, destroying two triangles per deleted edge. There can be no triangle-free subgraphs with fewer edge deletions.

Based on this, we can show that the perfect matchings of $G$ are in one-to-one correspondence with the $11n/2$-edge maximum triangle-free subgraphs of $\snub(G)$:
\begin{itemize}
\item Suppose we are given a perfect matching in $G$. If $u$ and $v$ are not matched, delete the cross edge for $t_u$ and $t_v$.  If $u$ and $v$ are matched, then for each of $t_u$ and $t_v$ delete the vertex triangle edge adjacent to both the $(u,v)$ cross and connecting edges.  After these deletions the graph $\snub (G)$ has lost $2n$ edges and is triangle free.
\item Suppose we are given a maximum size triangle free subgraph of $\snub(G)$.  Because there is a perfect matching in every regular bipartite graph, this subgraph must have exactly the size that a subgraph generated as above from a perfect matching would have: $11n/2$ edges in the subgraph after the deletion of $2n$ edges from $\snub(G)$. Therefore, the only edges that could have been deleted are cross and connecting edges.  Observe that if one connecting edge has been deleted then so to must its pair in order to destroy the remaining connecting triangle while continuing to destroy two triangles per deleted edge.  By similar reasoning, at most one edge from each vertex triangle could have been deleted.  Therefore the pairs of connecting edges that have been deleted determine a perfect matching.
\end{itemize}
Thus, this reduction is parsimonious, and \prob{\#TRI-FREE} is \cc{\#P}-complete.
\end{proof}

\begin{theorem}\label{thm:ergm-p-sharp}
Unless \cc{P}${}={}$\cc{\#P}, there can be no polynomial time algorithm for
\prob{ERGM-PARTITION}.
\end{theorem}
\begin{proof}
Again we take an instance of \prob{\#TRI-FREE} and use the $\trifreeergm (G,\alpha)$ with $\alpha = {n \choose 2} + 1$ this time.  By \autoref{lem:ergm-digits}, if $d_i$ is the number of triangle free subgraphs with $i$ edges, then the integer part of $Z$ can be written in base $2^\alpha$ as $d_{n \choose 2}\dots d_2d_1d_0$.  The sum of the digits from $d_k$ to $d_{n\choose 2}$ is equal to the number of triangle free subgraphs with greater than or equal to $k$ edges. Thus, if we could compute $Z$ we could determine this number in polynomial time, implying that \cc{P}${}={}$\cc{\#P}.
\end{proof}

\begin{theorem}\label{thm:ergm-drawing}
Unless $\cc{RP}=\cc{NP}$, there is no randomized polynomial algorithm for generating random graphs whose output distribution is within total variation distance $1-o(1)$ of a given ERGM with variable weights.
\end{theorem}
\begin{proof}
Suppose there was such an algorithm, $A$.  Then for some $G$ and $k$, an instance of \prob{TRI-FREE}, we can construct the corresponding $\trifreeergm(G,\alpha)$ with $\alpha = {n \choose 2}+1$.  Then running $A$ on this ERGM, we can verify whether or not the outputted graph is a triangle free subgraph with at least $k$ edges, output $1$ if the verification succeeds, or output $0$ if it fails.  If $G$ has no $k$-edge triangle-free subgraph then this procedure will output one with $o(1)$ probability.  On the other hand if there is such a subgraph, then the valid subgraphs have a $1-o(1)$ fraction of the ERGM's probability and with high probability $A$ will output one of them.  Thus if there is such an algorithm to generate random graphs then we can use it to show that $\cc{RP}=\cc{NP}$.
\end{proof}

\section{Dichotomy}

In this section we prove a dichotomy theorem that describes the hardness of computing ERGM partition functions when the set of features only consists of subgraph indicators. Suppose we are given a set $S$ of graphs, isomorphic copies of which are to be used as (heterogeneous) features of an ERGM. Then, as we show, if all graphs in $S$ have at most two vertices, then the partition function $Z$ of all ERGMs using $S$ as features can be computed in polynomial time. However, if $S$ contains a graph with three or more vertices, then it is $\cc{\#P}$-hard to compute $Z$ for ERGMs using features from $S$. This demonstrates that our results on the hardness of ERGMs are not specific to triangle features: the hardness results that we have proven using triangles cannot be avoided by replacing the triangles by a more clever choice of complex subgraph features.

We prove this result using a feature replacement strategy, in which we use a combination of heterogeneous weights on larger subgraphs to simulate weights on smaller subgraphs.
In this way, hardness results for all larger features follow from basic hardness results for two three-vertex features: triangles (the main feature for our previous hardness results) or three-vertex paths.

For this section we restrict the features for the ERGMs to indicator and count functions for specific subgraphs from a set $S$.  Without loss of generality we can assume that the ERGMs under consideration include a feature for each subgraph isomorphic to a graph in $S$, as missing features can be handled by giving them weight~$0$.

\begin{problem} \prob{$S$-ERGM-PARTITION} for a given set $S$ of graphs

\textsc{Input:} an ERGM, $E$, whose features are subgraph indicators and counts of graphs in $S$

\textsc{Output:} the partition function, $Z$, for $E$
\end{problem}

\begin{lemma}\label{lem:ergm-two-easy}
Let $S$ be a set of graphs to be used to define the features for an ERGM.  If $S$ contains only graphs on two or fewer vertices, then \prob{$S$-ERGM-PARTITION} can be solved in polynomial time.
\end{lemma}
\begin{proof}
The only nontrivial graphs that $S$ can contain are $K_2$ (a single edge) and its complement $\overline{K}_2$ (the empty graph on two vertices).  Without loss of generality we assume that $S$ contains both of these graphs.  Given an ERGM $E$, and a potential edge, $(i,j)$, we define the $x_{(i,j)}$ to be the sum of two terms: the weight of the indicator function that tests whether $i$ and $j$ induce a $K_2$ subgraph, and the weight of the count function for $K_2$.
Symmetrically, we define $\overline{x}_{(i,j)}$ to be the sum of two terms: the weight of the indicator function that tests whether $i$ and $j$ include a $\overline{K}_2$ subgraph, plus the weight of the count function for $\overline{K}_2$.  Then
\[
Z = \sum_G \prod_{e \in G} 2^{x_{e}} \prod_{e \notin G}2^{\overline{x}_{e}}.
\]
To compute this value, first compute $Y = \prod_{(i,j)} 2^{\overline{x}_{(i,j)}}$ and, for each edge $e = (i,j)$, define $x'_{e} = x_{e} - \overline{x}_{e}$.  Then
\[
Z/Y = \sum_G \prod_{e \in G} 2^{x'_{e}} = \prod_{(i,j)} \left(1 + 2^{x'_{(i,j)}}\right).
\]
Thus if $S$ contains only graphs on two vertices, we can compute $Z$ in polynomial time by computing $Y$ and $Z/Y$.
\end{proof}

\begin{definition}
For a graph $H$ with $k$ vertices, $H'$ an induced subgraph of $H$ with $k'$ vertices, an ERGM $E$, and a weight $\gamma$, we define the \emph{feature replacement} of $H'$ with $H$ in $E$ to be a new ERGM, defined as follows.  The vertex set of the new ERGM will include all the vertices of $E$ together with some new vertices added in the construction. If $E$ includes features that are not isomorphic to $H'$, these same features continue to exist in the new ERGM. For each indicator feature of a subgraph in $E$ that is isomorphic to $H'$, we perform the following steps:

\begin{enumerate}[Step 1:]
\item Add $k$ additional vertices to the new ERGM and pair $k'$ of them up with the $k'$ vertices of the subgraph of type $H'$.
\item Label the $k$ new vertices of the ERGM $l_1$, $l_2$, $\dots$, and $l_k$ and the $k'$ original vertices $l'_1$, $l'_2$, $\dots$, and $l'_{k'}$ such that the pair of $l'_i$ is $l_i$. Label the vertices of $H$ in the same way.
\item Add a new indicator feature, with weight $2\gamma$, that matches subgraphs isomorphic to $H$ that are induced by the set of vertices $\{l_1, l_2, \dots l_k\}$ with the numbering of these vertices matching the numbering of the isomorphic copy of~$H$. Here $\gamma$ is a parameter to be specified later.
\item For each $i$ from $1$ to $k'$, add another indicator feature, with weight $2\gamma$, 
matching subgraphs isomorphic to $H$ that are induced by the set of vertices 
$\{l_1, l_2, \dots, l_{i-1}, l'_i, l_{i+1}, \dots l_k\}$.
\item Add one more indicator feature, with the same weight as indicator feature from $E$, 
matching subgraphs isomorphic to $H$ that are induced by the set of vertices
$\{l'_1,l'_2,\dots l'_{k'},l_{k'+1},\dots l_k\}$, i.e. with all the first $k'$ vertices swapped out for their pair in the original vertices.
\end{enumerate}

To handle features that count the number of subgraphs isomorphic to $H'$, we run the same process for each of the $k'!{n\choose k'}$ possible induced subgraphs of this type.
\end{definition}

\autoref{fig:featurereplacement} shows this feature replacement process for replacing $K_3$ with a wheel graph on $6$ vertices.

\begin{figure*}[t]
\centering\includegraphics[width=0.85\textwidth]{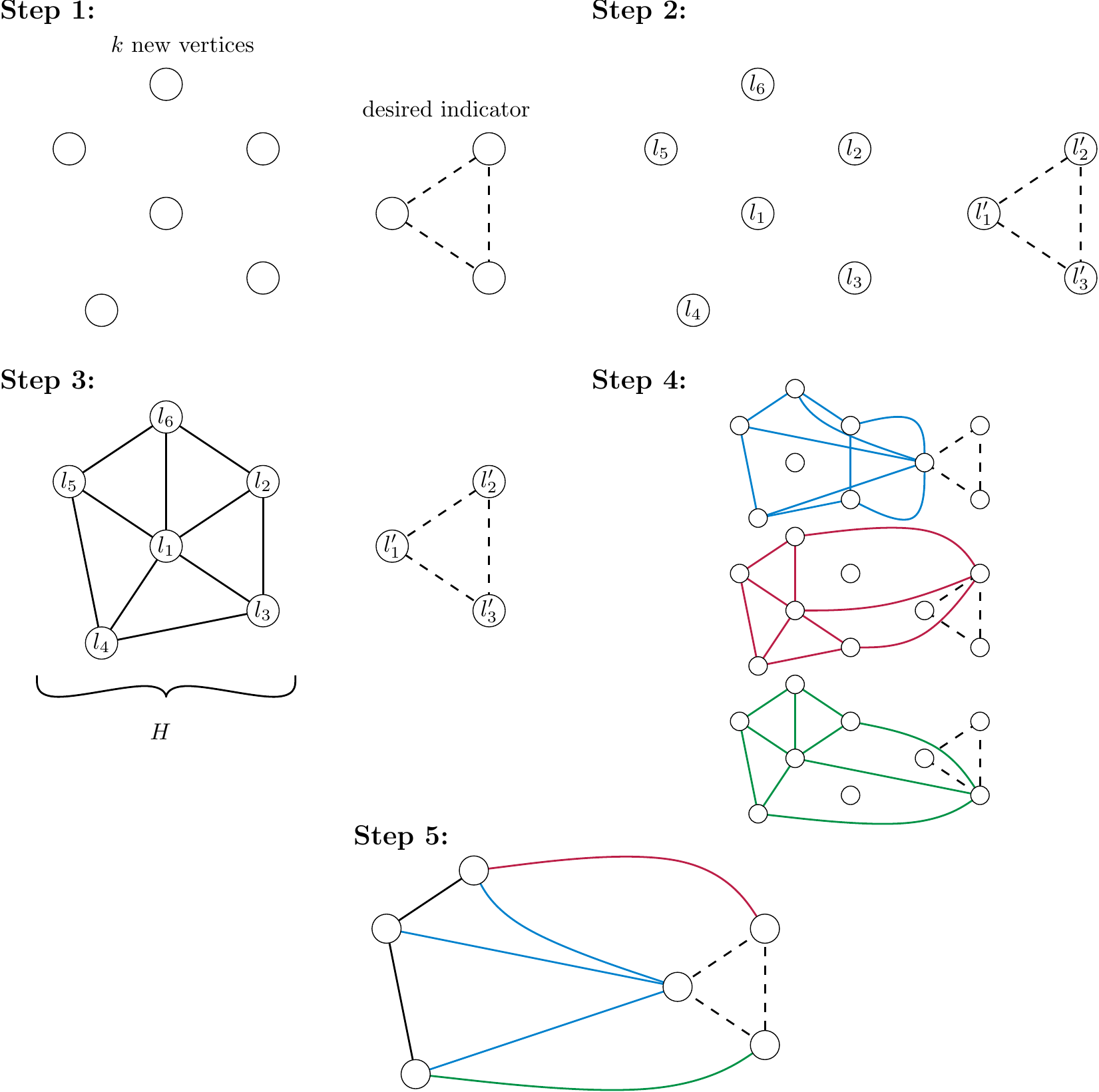}
\caption{An example of feature replacement replacing a triangle indicator with six new vertices and five wheel graph indicators.}
\label{fig:featurereplacement}
\end{figure*}

\begin{lemma}\label{lem:feature-replace}
Given two graphs $H$ and $H'$, where $H'$ is an induced subgraph of $H$, the partition function of an ERGM that has only integer weights and uses $H'$ can be computed from the partition function of an ERGM that uses $H$ instead of $H'$.
\end{lemma}
\begin{proof}
First let $w^+$ be the sum of all positive weights in the ERGM that uses $H'$ and $w^-$ be the absolute value of the sum of all negative weights.  We know the digits of $Z$ for this ERGM are within $w^-$ digits of the right of the decimal point and $w^+$ digits to the left of the decimal point.

Use feature replacement to replace $H'$ with $H$ using a weight of $\gamma = \left(w^+ + w^-\right)$.

The ERGM obtained from feature replacement has polynomially many new vertices and indicator features. In this ERGM, there exist states in which all of the indicator features with weight $2\gamma$ are true; let $s$ be the number of these indicator features. Then for each of these states, the density of the graph will include a factor of $2^{2s\gamma}$ for those features, while all the other states will omit at least one factor of $2^{2s\gamma}$. Thus, by looking at the binary digits of the partition function extending from the $2^{2s\gamma}$ bit upwards (as seen in \autoref{fig:baseshift}), we can recover the subset of the partition function generated only by the states in which these indicator features are all true.
For these states, the remaining terms in the weight of each state coincide with the corresponding terms in the weights of the states of the original ERGM.
\end{proof}

\begin{figure*}[t]
\centering\includegraphics[width=0.95\textwidth]{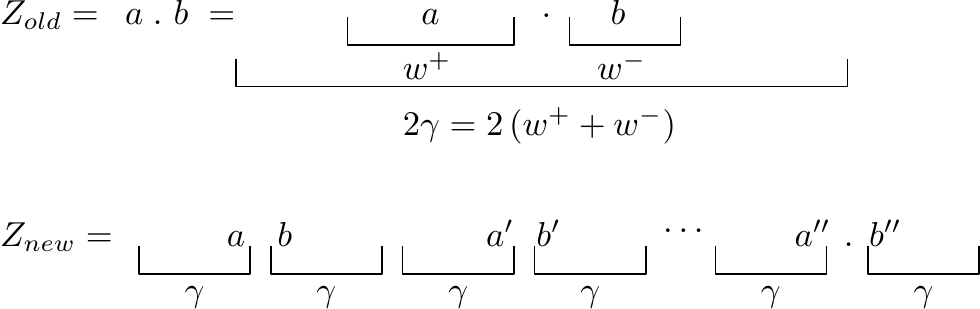}
\caption{$Z_{old}$ is found in the leading two digits of $Z_{new}$ in base $2^\gamma$.}
\label{fig:baseshift}
\end{figure*}

\begin{lemma}\label{lem:ergm-three-hard}
Let $S$ be a set of graphs containing any graph $H$ on three or more vertices. Then \prob{$S$-ERGM-PARTITION} is \cc{\#P}-hard.
\end{lemma}
\begin{proof}
$H$ must contain at least one of $K_3$, $P_2$, $\overline{P_2}$, or $\overline{K_3}$ as an induced subgraph; we handle each case separately. The cases of $\overline{P_2}$ or $\overline{K_3}$ can be transformed into the cases of $P_2$ or $K_3$ by complementing all of the features used in the ERGM (keeping the weights the same), which produces an ERGM whose probability on any graph is the same as the probability of the original ERGM on the complementary graph. In particular, this transformation leaves the partition function unchanged. Thus, we need only consider the cases of $K_3$ and $P_2$.

If $H$ contains $K_3$ as an induced subgraph, we proceed in the same manner as \autoref{thm:ergm-p-sharp}. However, we have to simulate the weights on triangle and edge subgraphs using indicator features for copies of~$H$.  To do so we observe that because $H$ contains $K_3$ it also contains $K_2$ and so two applications of \autoref{lem:feature-replace} allow us to reduce the instance of \prob{\#3REG-BIP-MATCH} to an ERGM using subgraph indicator features of graphs in $S$.

If $H$ contains $P_2$ as an induced subgraph, we instead reduce from \prob{\#MATCHING}.  Given a bipartite graph, $G$, as input, we create for each edge in $G$ an indicator feature for that edge with weight ${n \choose 2}$.  For edges not in $G$ and any $P_2$ in $G$, we create an indicator feature for that subgraph with weight $\beta = -{n \choose 2}^2 - {n \choose 2} - 1$. By an argument similar to \autoref{lem:ergm-digits}, if $d_i$ is the number of matchings of $G$ with $i$ edges, then the partition function in base $2^{n\choose 2}$ is $d_{n\choose2}\dots d_2d_1d_0$.  Thus $d_n$,the $n$-th digit of $Z$, counts the number of perfect matchings.  Now using \autoref{fig:featurereplacement} we can reduce this ERGM using $P_2$ and $K_2$ to another ERGM using $H$.
\end{proof}

\begin{theorem}\label{thm:ergm-dichotomy}
Given a set of subgraphs, $S$.  If $S$ contains a graph on three or more vertices, \prob{$S$-ERGM-PARTITION} is \cc{\#P}-hard and can be computed in polynomial time otherwise.
\end{theorem}
\begin{proof}
The result follows from \autoref{lem:ergm-two-easy} and \autoref{lem:ergm-three-hard}.
\end{proof}

\section{Conclusion}

We have shown \prob{ERGM-PARTITION} to be \cc{\#P}-hard and inapproximable via reductions that are very close to the natural usage of ERGMs.  Additionally, we showed that the hardness of ERGM partition function computation can be classified by their subgraph features where ERGMs that use subgraphs with more than two vertices lead to hard to compute partition functions.  All but the simplest of ERGMs are fundamentally difficult to work with and if the ability to sample from the distribution is required, then different distributions on graphs are necessary.

\newpage
{\raggedright
\bibliographystyle{abuser}
\bibliography{paper}}

\end{document}